\documentclass[preprint]{iacrtrans}
\makeatletter
\renewenvironment{abstract}{%
  \small\quotation\setlength{\parindent}{0pt}\noindent
  \textbf{\textsf{Abstract.}}}
  {\endquotation}
\makeatother

\usepackage{fancyhdr}
\usepackage[utf8]{inputenc}
\usepackage[dvipsnames]{xcolor}
\usepackage{graphicx}
\usepackage{subcaption}
\usepackage{hyperref}
\usepackage{graphicx} 
\usepackage{amsmath}
\usepackage{amssymb}
\usepackage{physics}
\usepackage{enumitem}
\usepackage{cleveref}
\usepackage{booktabs}
\usepackage{hyperref}
\usepackage{ dsfont }
\usepackage{stmaryrd}
\usepackage{mathtools}
\usepackage{algorithm}
\usepackage{algorithmic}
\usepackage[normalem]{ulem}
\usepackage[export]{adjustbox}

\usepackage{tcolorbox}

\usepackage{mdframed}

\usepackage[colorinlistoftodos]{todonotes}

\newcommand{\M}{\mathcal{M}}

\newcommand{\F}{\mathcal F}
\newcommand{\G}{\mathcal G}
\renewcommand{\P}{\mathcal P}
\newcommand{\eps}{\epsilon}

\newcommand{\D}{\mathcal D}

\renewcommand{\S}{\mathcal S}

\newcommand{\SIM}{\mathcal {SIM}}

\newcommand{\SD}{\mathcal {SD}}
\newcommand{\PP}{\mathcal {PP}}

\newcommand{\R}{\mathcal R}

\newcommand{\probP}{\text{I\kern-0.15em P}}

\newcommand{\kb}[1]{\ket{#1}\bra{#1}}

\newcommand{\bea}{\begin{eqnarray}} %
\newcommand{\eea}{\end{eqnarray}}

\newtheorem{assumption}{Assumption}

\Crefname{subassumption}{Assumption}{Assumptions}
\Crefname{assumption}{Assumption}{Assumptions}

\Crefname{subassumptionTwo}{Assumption}{Assumptions}
\Crefname{proposition}{Proposition}{Propositions}
\Crefname{property}{Property}{Properties}

\title{On the cryptographic potential of single-qubit rotations\\}

\author{Alex B. Grilo\inst{1} \and Lucas Hanouz\inst{1,2,*} \and Anne Marin\inst{2}}
\institute{
Sorbonne University, CNRS, LIP6, Paris, France
\and
VeriQloud, Paris, France
}

\date{March 2025}

\begin{document}
\maketitle

\vspace{-1.2cm}
\begin{center}
    \small{$^{*}$ Corresponding author: lucas.hanouz@lip6.fr}
\end{center}

\begin{abstract}

In the domain of quantum communication, cryptographic protocols often require users to have access to trusted qubit sources or detectors. Recently, it was shown that on an architecture called the Qline, several protocols can equivalently be performed by parties capable only of single-qubit rotations. 

In this work, we introduce two composably secure constructions that together show how in most quantum cryptographic protocols, parties traditionally required to perform trusted qubit preparation or measurement can delegate these tasks to an untrusted provider and instead rely on a trusted single-qubit rotation device.

Our first construction implements single-qubit measurement and is universally applicable across any context. In contrast, our second construction, which addresses qubit preparation, relies on specific assumptions regarding the underlying protocol. We show, however, that these assumptions are inherently satisfied by the vast majority of common quantum cryptographic protocols.

A notable consequence of our results is the formal validation of the Qline as a versatile architecture capable of supporting a wide range of single-qubit protocols.
\end{abstract}

\section{Introduction} \label{sec:intro}

Quantum cryptography represents a paradigm shift in secure communications, leveraging the principles of quantum mechanics to relax assumptions required in classical cryptography or enable entirely new cryptographic paradigms.

This field effectively arose from the work of Bennett and Brassard on Quantum Key Distribution (QKD)~\cite{BB84}. It showed that using quantum resources, two distant authenticated parties can agree on a shared secret bit string without relying on any computational assumption, a task that is provably impossible using only classical resources. By transmitting information encoded into a quantum state of a physical system, one can exploit the fact that any attempt to gain knowledge on these states inevitably disturbs them, thereby establishing a communication channel where any eavesdropping attempts can be detected.
This paradigm has paved the way for more advanced primitives, establishing QKD not just as a specific application, but as the catalyst for the entire domain of quantum-secure communication.

Quantum resources are, however,  technologically challenging to manipulate. The hardware required to prepare, transmit, modify, or measure quantum states remains expensive and imperfect; furthermore, current technology cannot properly handle high-dimensional entangled states or reliably store practical quantities of quantum information over meaningful timescales.

By far the most practical and least challenging quantum cryptography protocols to implement are the so-called \textit{prepare-and-measure} (PM) protocols. In prepare-and-measure scenarios, analogous to the initial proposition of QKD, users are only required to either prepare single-qubit quantum states and immediately transmit them, or receive single-qubit states and immediately measure them.
Yet, even these less stringent conditions give rise to serious challenges when realistic, full-scale networks are considered. In particular, the cost of the required hardware can become prohibitive.

To address these challenges and enhance the connectivity of quantum communication networks, a specific architecture called Qline has been proposed~\cite{doosti-hanouz23establishing}. The Qline consists of the simplest prepare-and-measure setup where a single-qubit source and detector are linked, but with intermediate nodes added in between. These nodes are limited to performing single-qubit rotations and are not required to handle state preparation or measurement.

The Qline has successively been shown to support several interesting cryptographic protocols. Clementi \textit{et al.} demonstrated a Quantum-enhanced Classical multiparty computation protocol on the Qline \cite{CPEW17}. Polacchi \textit{et al.}~\cite{polacchi2023multi} introduced a protocol for secure multi-client delegated quantum computing for a Qline connected to a quantum computer. Doosti \textit{et al.}~\cite{doosti-hanouz23establishing} showed that any pair of players can establish symmetric keys with the same level of security as QKD, provided that the end-nodes are trusted, an assumption which was later removed by Grilo, Hanouz and Marin~\cite{SecretSharingQline} by showing the security of the more general task of establishing shares of additive secret sharing.

Beyond the specific scope of the Qline architecture, devices capable of performing single-qubit rotations have been studied in broader contexts such as for delegating~\cite{Qenclave} and verifying~\cite{KLMO24} quantum computation.

Collectively, these results suggest that the Qline possesses the versatility to support an even broader range of applications.
Furthermore, the fact that intermediate users on the Qline can perform cryptographic tasks without requiring trusted qubit sources or detectors raises a question of independent interest: what is the cryptographic potential of single-qubit rotation devices, and which protocols do they enable? This work addresses these questions and provides the following insight: almost any single-qubit protocol involving a trusted device for qubit preparation or measurement can instead be performed using an untrusted version of that device coupled with a trusted single-qubit rotation device.

\medskip

Our main contributions consist of two constructions, tailored respectively to qubit sources and detectors:
\begin{enumerate}
    \item The first construction, presented in~\Cref{sec:Detector}, implements a trusted single-qubit measurement from a trusted single-qubit rotation device followed by a (possibly remote) untrusted detector.
    \item The second construction, detailed in~\Cref{sec:Source}, describes some sufficient conditions for cryptographic protocols for which we can replace a trusted qubit source by a (possibly remote) untrusted single-qubit source and a trusted single-qubit rotation device.
\end{enumerate}

The sufficient conditions of protocols that we consider in our second contribution are quite extensive and encompass in particular several known cryptographic protocols. Specifically, they include protocols that admit an entanglement-based version that remains secure when the entangled states are provided by an untrusted party. We provide a detailed analysis in \Cref{sec:examples} showing that most common protocols for QKD and Quantum Oblivious Transfer satisfy these properties.

Conceptually, our constructions can be understood as compilers turning single-qubit quantum cryptographic protocols that require the transmission or reception of single qubits into equally secure variants where the user instead needs a single-qubit rotation device, while the tasks of qubit preparation and measurement are delegated to untrusted parties. The user also requires additional local, private randomness.

We carry out our security analysis within the Abstract Cryptography framework~\cite{AbstractCrypto}, which endows our constructions with composability. This property is crucial for the theoretical coherence of our framework, as it ensures that our results can be seamlessly integrated into broader cryptographic constructions while preserving the soundness of existing security proofs.

The main consequence of our work is a general alternative for measurement and preparation implementation, in most single-qubit quantum cryptographic protocols.  
By shifting the requirement of trust from preparation and measurement devices to rotation units, our framework allows for a functional "standardization" of user hardware. In this model, while physical sources and detectors remain necessary for the protocol's execution, they no longer need to be trusted. Consequently, these components could be provided by a third-party quantum-communication service provider, while the users maintain security through standardized rotation-based devices that remain invariant across a wide range of applications.

A direct implication of our result is the extension of most single-qubit protocols to the Qline architecture. Specifically, our framework allows any intermediate node to participate in these protocols as if they possessed a trusted source or detector, thereby significantly expanding the functional capabilities of the Qline. 

We must emphasize, however, that our construction still presupposes a trusted quantum device, which contrasts with the traditional definition of device-independent protocols that typically eliminate the need for any quantum-level trust. 
Moreover, our theorems do not avoid side-channel attacks, and the single-qubit rotation devices must satisfy assumptions comparable to those required for qubit sources and detectors in standard protocols. 
In the case of single-qubit rotation devices, this consists of the implicit assumption that the rotations are guaranteed to be applied on single qubits, as opposed to higher-dimensional systems.

We note that while the manufacturing of qubit rotation devices is technically mature (as demonstrated by photonic Qline implementations \cite{Qline_implementation_berlin,Schmid_2005}), to the best of our knowledge, no implementation has yet been reported with the additional guarantee that the device strictly processes single-qubit states. At present, the technological overhead and cost associated with ensuring this property remain to be fully characterized. We therefore position this work as a theoretical motivation for exploring the feasibility of experimental realization of such "certified" single-qubit rotation devices, which would serve as a versatile foundation for secure quantum communication.

\section{Preliminaries} \label{sec:prelim}

\subsection{Notation}

We assume basic knowledge of the theory of quantum communication and computing.

We write $\ket{+_\theta}$ the state $\frac1{\sqrt2}\big( \ket0 + e^{i\theta}\ket1 \big)$ for any $\theta\in\mathds{R}$.

We write $R_Z(\theta)$ the single-qubit rotation $\begin{bmatrix} 1&0\\0&e^{i\theta} \end{bmatrix}$
of angle $\theta\in\mathds{R}$ around the $Z$ axis of the Bloch sphere.

For any $\theta\in\mathds R$, we refer to the measurement basis of eigenstates $\{\ket{+_{\theta}}, \ket{+_{\pi+\theta}}\}$ as a measurement in basis $\theta$. A measurement in this basis yields the classical outcome $0$ if $\ket{+_{\theta}}$ is measured, and 1 for $\ket{+_{\theta+\pi}}$.

\subsection{The Abstract Cryptography framework} \label{sec:AC_presentation}

We prove the security of our protocol using the \emph{Abstract Cryptography} framework~\cite{AbstractCrypto}.

In this framework, cryptographic protocols are defined as systems: abstract objects with \emph{interfaces} that define all possible inputs and outputs of the system. Each interface represents an entity's access to the system. A cryptographic construction typically includes a \emph{user} interface, where the interaction between the honest user and the system occurs, as well as an adversarial interface, called the \emph{outer} interface, which encapsulates the attacker's capabilities.
        
Systems can be composed, either in parallel or sequentially.
The parallel composition of two systems $\R$ and $\S$, denoted $\R||\S$, is a system with the interfaces of both sub-systems. 
It simply describes the fact that these systems are put side by side and seen as a whole, unique system.

The sequential composition describes the fact that the output of a system can be used as input by other systems. For instance, two systems $\R$ and $\S$ can be sequentially composed \emph{at the interfaces $i_\R$ of $\R$ and $j_\S$ of $\S$} if each input (respectively output) of these interfaces can be associated with a unique output (respectively input) of the other interface.
When it is clear at which interfaces a sequential composition occurs, we denote it  $\R\circ \S$ or simply $\R\S$ without specifying the interfaces $i_\R$ and $j_\S$.
The resulting system has all the interfaces of both sub-systems except for $i_\R$ and $j_\S$.

In this framework, the security of a cryptographic scheme is defined in a simulation-based fashion, as the ``closeness'' of that system to an ideal version of it. 
Formally, a protocol $P$ of ideal version $\tilde P$ is said to be \emph{$\eps$-secure} if there exists a system $\SIM$ (called a simulator) such that $P \approx_\eps \tilde P\circ \SIM$, with $\approx_\eps$ denoting indistinguishability, for any computationally unbounded entity and up to an advantage at most $\eps$.

Such a definition of security is said to be {\em composable}, meaning that any $\eps$-secure system can replace its ideal version in any setting with no discernible effect for adversaries with abilities comparable to the distinguisher, except with probability $\eps$ (see~\cite{AbstractCrypto}, theorem $2$).

\section{Implementation of a qubit detector} \label{sec:Detector}

In this section, we show that in any protocol involving single-qubit measurements in a given plane, these trusted measurements can be securely replaced by a device that performs single-qubit rotations in that same plane followed by a measurement by an untrusted detector.

We first define the context and a protocol involving single-qubit rotations, and then prove \Cref{thm:detector} which states that this protocol securely implements the task of a detector.

For convenience, we restrict our analysis to the case of measurements in the $XY$ plane of the Bloch sphere, but we point out that, due to equivalence up to a basis change, our result applies to any other plane.

\medskip

Let $P$ be a given cryptographic protocol involving at least one party, that we will call Bob (or $B$). Let $P^B$ denote the local protocol that $B$ follows with the following structure:

\begin{assumption} \label{ass:pb_structure}
    $P^B$ can be decomposed in two distinct phases:
    \begin{enumerate}
        \item \textbf{State distribution:}  $B$ receives $N$ qubits. They then sample $(\theta_n)_{n\in[N]} \in [0, 2\pi[^N$ according to a given probability density function $p_\Theta$ and measure each qubit $n\in[N]$
        in basis $\theta_n$. We write $(x_n^B)_{n\in N}$ their classical outcomes.
        \item \textbf{Classical post-processing:} $B$ performs an arbitrary protocol with inputs $(\theta_n)_{n\in N}$ and $(x_n^B)_{n\in N}$.
    \end{enumerate}
\end{assumption}

Let $\P^B$ be a system that formalizes $P^B$. \Cref{ass:pb_structure} implies that $\P^B$ can be decomposed as $\P^B = \SD^B \circ \PP$ with $\SD^B$ and $\PP$ systems that respectively implement the state distribution and post-processing steps. We formally define $\SD^B$ below, from the structure imposed by \Cref{ass:pb_structure}.

\begin{definition}[$\SD^B$, $(\SD^B_n)_{n\in[N]}$]
    $\SD^B$ is represented \Cref{fig:SDB}. It is a system formalizing $P^B$'s state distribution. Upon receiving at its outer interface a quantum state $\rho$ of dimension $2^N$, $\SD^B$ samples $(\theta_n)_{n\in[N]} \in [0, 2\pi[^N$ according to $p_\Theta$ and measures each qubit $n\in[N]$ of $\rho$ in the basis $\theta_n$. $\SD^B$ then outputs $(\theta_n)_{n\in[N]}$ along with the measurement outcomes $(x_n)_{n\in[N]}$ at its user interface.
    Note that the measurement step can be decomposed as a succession of single-qubit measurements. Formally, 
    \begin{equation} \label{eq:SDBn_def}
        \SD^B = \G_\Theta \circ {\big|\big|}_{n\in[N]}\SD^B_n,
    \end{equation}
    where $\G_\Theta$ is a system that samples $(\theta_n)_{n\in[N]}$ and where $(\SD^B_n)_{n\in[N]}$ are single-qubit measurement systems.
\end{definition}

Let Charlie be a party equipped with a single-qubit rotation device.
As discussed, our goal is to achieve a protocol enabling Charlie, with the help of an untrusted detector, to implement Bob's state distribution $\SD^B$. For that, we define below the system $\SD^C$ representing the state distribution of Charlie. $\SD^C$ formalizes the following intuition: to measure an unknown state in a given basis $\theta$, Charlie may simply rotate the input state by an angle $-\theta$ and ask the untrusted detector to measure the resulting qubit in the Hadamard basis and to announce the classical outcome. On top of this, to hide the outcome of the measurement, Charlie randomly flips the qubit by changing the rotation angle to $-\theta +r\pi$ with a uniformly random $r\in\{0, 1\}$, effectively encrypting the classical outcome obtained by the detector in a way that ensures that only Charlie can read the outcome\footnote{Note that while an obvious attack on the correctness would be for the detector to lie about the classical outcome of the measurement, an identical attack can be performed on the trusted detector of Bob by simply flipping the qubit before providing the qubit to Bob. Such an attack lies within the threat model considered by quantum cryptographic protocols as the quantum channel is most conventionally assumed to be controlled by the adversary.}.

\begin{figure}[!ht]
    \centering 
    \includegraphics[width=0.95\textwidth]{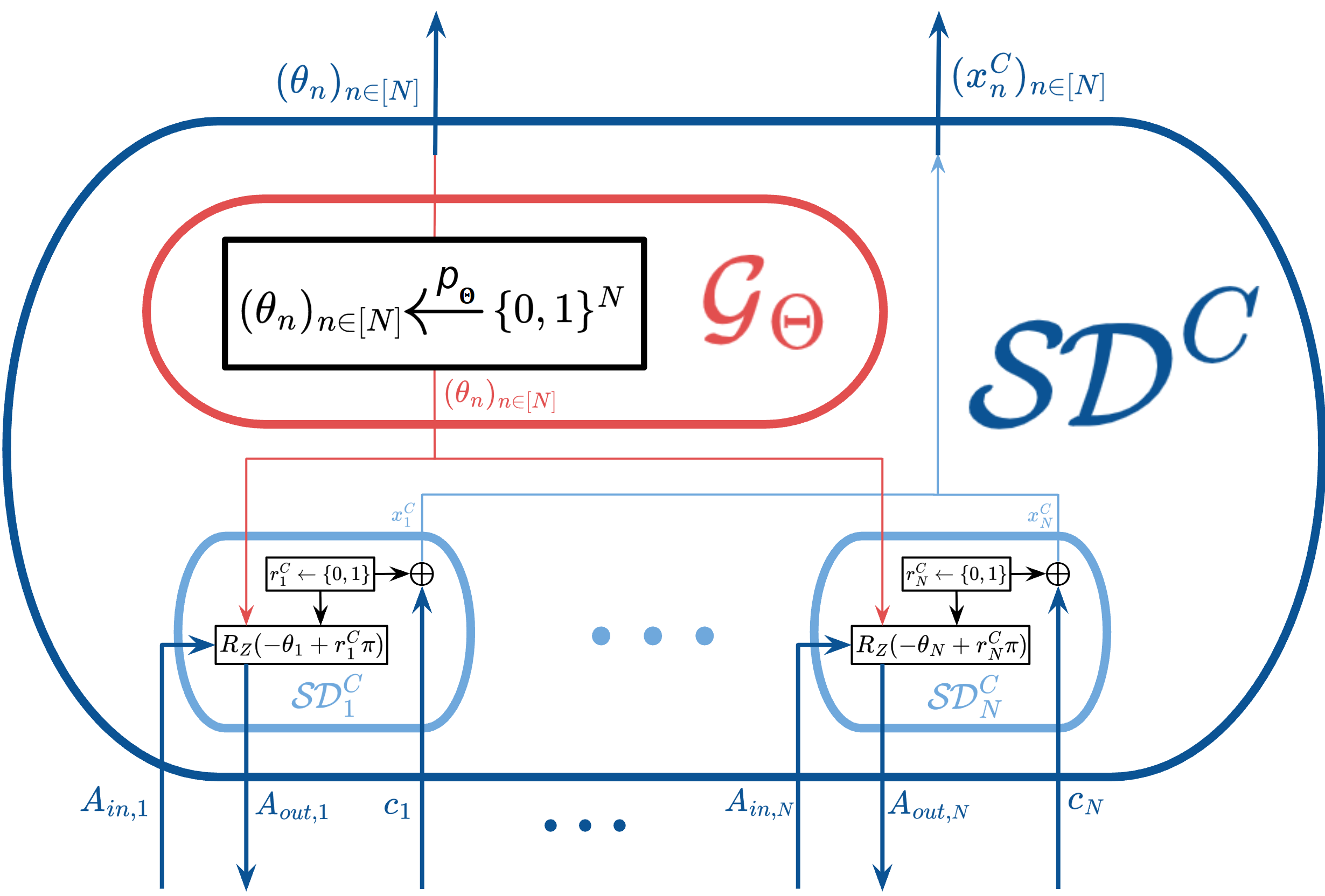}
    \caption{\centering The system $\SD^C$ of Charlie.}
    \label{fig:SDC}
\end{figure}
\begin{figure}[!ht]
    \centering 
    \includegraphics[width=0.95\textwidth]{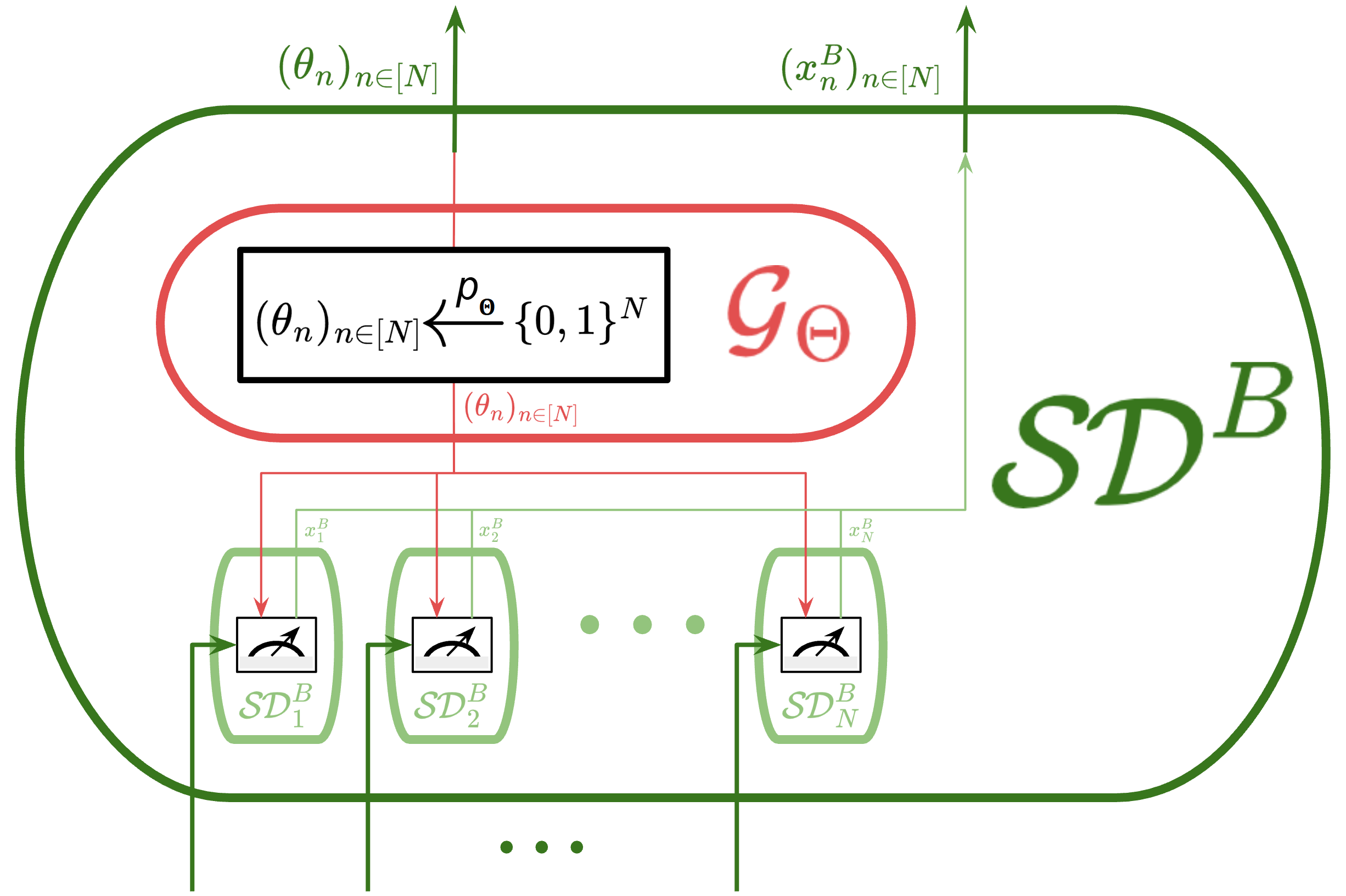}
    \caption{\centering The system $\SD^B$ of the detector Bob.}
    \label{fig:SDB}
\end{figure}

\begin{definition}[$\SD^C$, $(\SD^C_n)_{n\in[N]}$]
    $\SD^C$ is represented \Cref{fig:SDC}.
    $\SD^C$ samples $(\theta_n)_{n\in[N]} \in [0, 2\pi[^N$ according to $p_\Theta$ as well as a uniformly random bit string $r^C = (r^C_n)_{n\in[N]} \in \{0, 1\}^N$. Then, successively for $N$ rounds $n\in[N]$, $\SD^C$ receives a single qubit, applies the $R_Z(-\theta_n+r_n^C\pi)$ operation to that qubit and forwards the resulting state to its outer interface. Finally, upon receiving at its outer interface $N$ bits $(c_n)_{n\in[N]}$, $\SD^C$ outputs at its user interface 
    \begin{align}
        \label{eq:tic_def}
        &(\theta_n)_{n\in[N]}, \text{ and } \\
        \label{eq:xic_def}
        &(x^C_n)_{n\in[N]} = (r^C_n \oplus c_n)_{n\in[N]}.
    \end{align}
\end{definition}

Note that the successive rounds can be seen as a parallel composition of single-qubit systems $(\SD^C_n)_{n\in[N]}$ that receive an angle $\theta_n$ as well as a qubit and return the qubit after the appropriate gate has been applied. As a consequence, 
\begin{equation} \label{eq:SDCn_def}
    \SD^C = \G_\Theta \circ {\big|\big|}_{n\in[N]}\SD^C_n.
\end{equation}

Finally, inspired by the construction of the entanglement-based protocol of \cite{SecretSharingQline}, we define the simulator $\SIM$, an abstract system used as a proof tool to show that $\SD^C$ securely implements $\SD^B$. More precisely, $\SIM$ is designed so that $\SD^C \approx_0 \SD^B\circ\SIM$, which we will prove later (see \Cref{thm:detector}). As a consequence, it has a first interface, which we will call "inner" interface, meant to connect to the outer interface of $\SD^B$ (with an $N$-qubits output), and an outer interface mirroring the one of $\SD^C$ (with two $N$-qubits signals: one input and one output).

\begin{definition}[$\SIM$, $(\SIM_n)_{n\in[N]}$]

    $\SIM$ is represented \Cref{fig:SB}. 
    Successively for $N$ rounds $n\in[N]$, $\SIM$ performs the following tasks:
    \begin{enumerate}
        \item $\SIM$ receives a single qubit in a register $A_n$, prepares a register $B_n$ in the state $\ket 0_{B_n}$ and applies the $CNOT_{A_n, B_n}$ gate before returning register $A_n$ to its outer interface.
        \item Upon receiving a bit $c_n$ (from its outer interface), $\SIM$ applies the $Z^{c_n}$ operation to register $B_n$ before outputting that register at its inner interface, for it to be received by system $\SD^B$.
    \end{enumerate}
    Note that according to this definition, $\SIM$ can be viewed as the parallel composition of $N$ subsystems that each perform the task of one round. We denote these subsystems by $(\SIM_n)_{n\in[N]}$:
    \begin{equation} \label{eq:SIMn_def}
        \SIM = {\big|\big|}_{n\in[N]}\SIM_n.
    \end{equation}
\end{definition}

\begin{figure}[!ht]
    \centering 
    \includegraphics[width=0.95\textwidth]{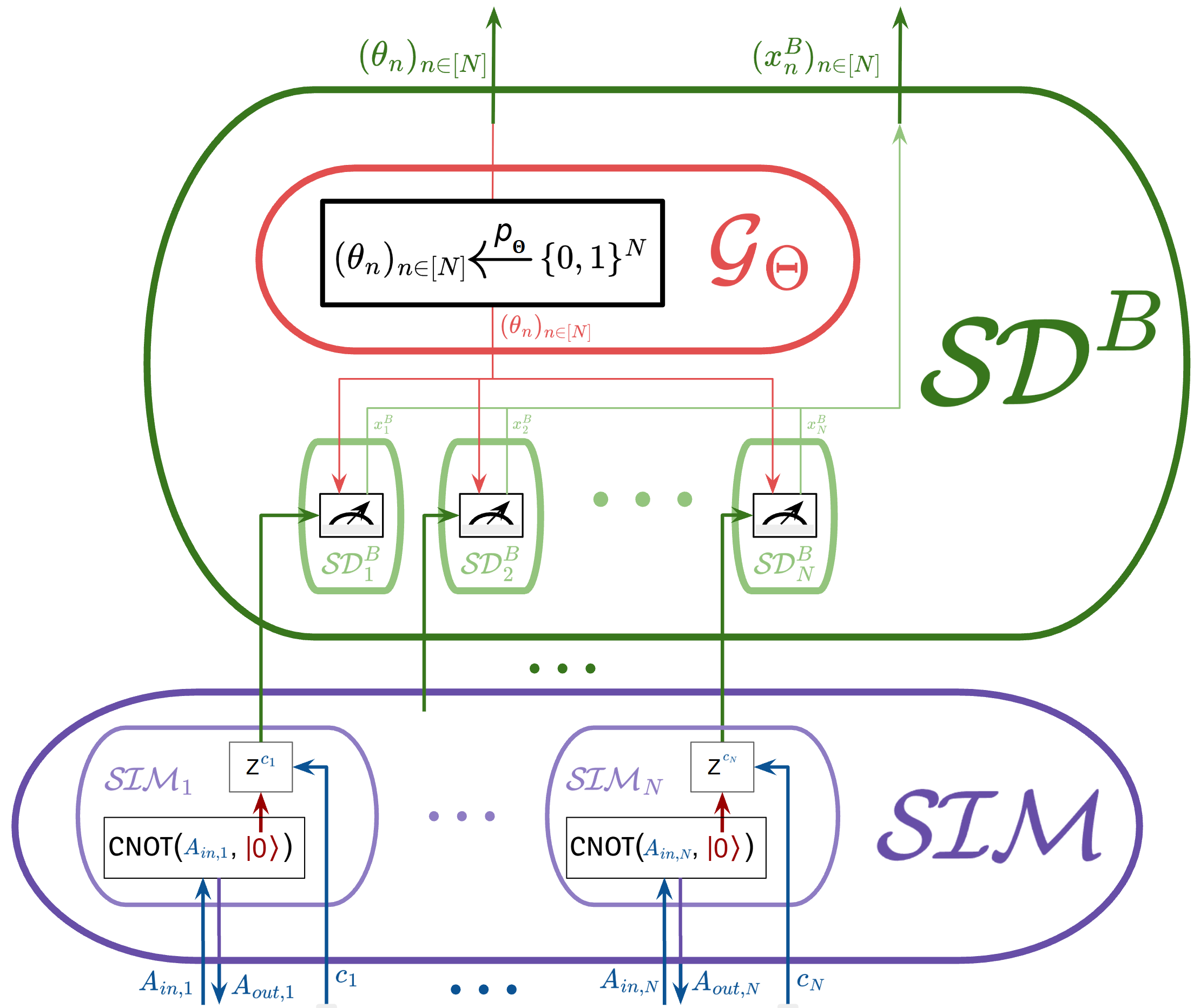}
    \caption{\centering The system $\SIM$ plugged on $\SD^B$.}
    \label{fig:SB}
\end{figure}

In order to prove the desired result $\SD^C \approx_0 \SD^B\circ\SIM$ (i.e., \Cref{thm:detector} below), which states the security of Charlie's protocol $\SD^C$, we first show the following \Cref{thm:detector_single_qubit_systems} which focuses on the single-qubit subsystems composing $\SD^C$ and $\SD^B\circ\SIM$, namely $\SD^C_n$, $\SD^B_n$, and $\SIM_n$.

\begin{lemma} \label{thm:detector_single_qubit_systems}
    Under \Cref{ass:pb_structure}, for all $n\in[N]$
    \begin{equation}
        \SD^C_n \approx_0 \SD^B_n \circ \SIM_n.
    \end{equation}
\end{lemma}

\begin{proof}
    This proof is inspired by the proof of Theorem 1 of \cite{SecretSharingQline}.
    
    Let $n\in[N]$ be fixed. 

    Let $\S$ be a system which is either $\SD^C_n$ or $\SD^B_n\circ\SIM_n$ with equal probability $\frac 12$, and consider a distinguisher $\D$, i.e., an unbounded system able to perform any operation allowed by the laws of quantum mechanics, and whose goal is to distinguish which of the two systems $\S$ is ($\D$ and $\S$ are depicted in  \Cref{fig:Sn}).
    
    Our goal now is to prove that $\D$ succeeds with a maximum probability of $\frac 12$ --- the probability of a random guess. 

    \begin{figure}[!ht]
        \centering 
        \includegraphics[width=0.95\textwidth]{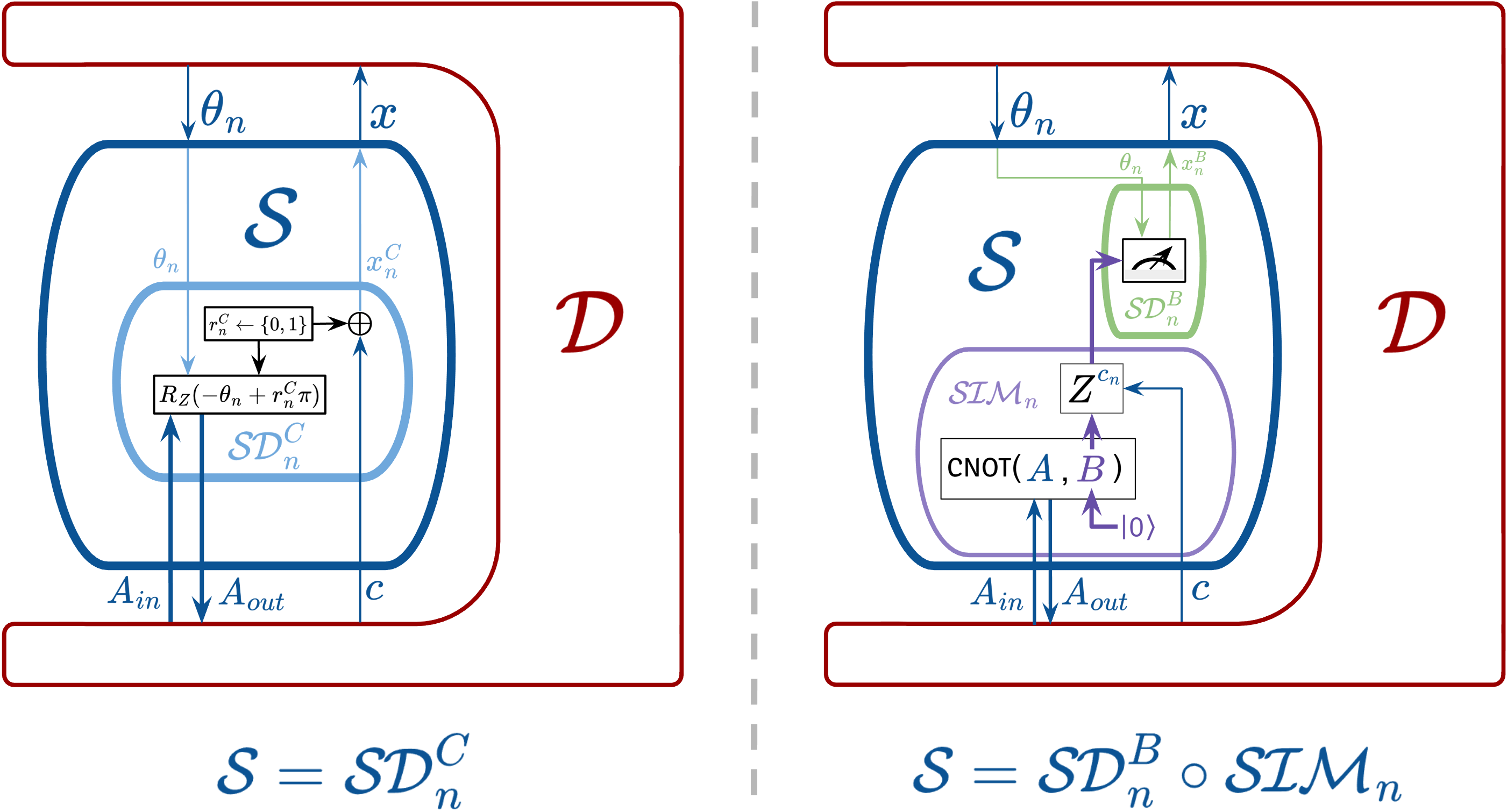}
        \caption{\centering The distinguisher $\D$ and the system $\S$, either equal to  $\SD^C_n$ or $\SD^B_n \circ \SIM_n$.}
        \label{fig:Sn}
    \end{figure}

    Without loss of generality, we can consider that the behavior of $\D$ amounts to: 
    \begin{enumerate}
        \item Choosing an angle $\theta_n$ and preparing a pure state $\ket\tau_{AD_1}$ in a single-qubit register $A$, and a private register $D_1$.
        \item Sending $\theta_n$ and register $A$ to $\S$ (through the input $\rho_{in}$ of $\S$) for it to apply its operation, and getting the register back (from the output $\rho_{out}$ of $\S$).
        \item \label{it:ci} Applying a measurement $\M^{(c)}$ on registers $A$ and $D_1$ with outcome $c$ and sending this classical bit to $\S$.
        \item Obtaining the final output $x_n$ of $\S$ and storing it in a register $D_2$, as well as the classical outcome $c$ of the previous measurement.
        \item Applying a final measurement $\M_{\D}$ on registers $A$, $D_1$, $D_2$, and outputting the single-bit measurement outcome $d$.
    \end{enumerate}
    
    \noindent We call $\psi^B$ (respectively $\psi^C$) the state of registers $A, D_1, D_2$ before step 5, on which the final measurement $\M_\D$ is performed if $\S = \SD^B_n\circ\SIM_n$  (respectively if $\S = \SD^C_n$). We will show that the trace distance of these two states $Tr(\psi^B, \psi^C)$ is $0$, thus proving that $\M_\D$ cannot distinguish the two states better than a random guess.

    We first suppose that $\S = \SD^B_n\circ\SIM_n$.
    We write $\ket\tau_{AD_1} = \alpha \ket{\tau_{0}}_{D_1} \ket{0}_{A} + \beta \ket{\tau_{1}}_{D_1} \ket{1}_{A}$, and
the state of registers $A, B, D_1$ after the $CNOT$ gate is

    \begin{align}
        &
        &&\mkern-191mu(\mathds{I}_{D_1} \otimes CNOT_{AB})
        (\ket\tau_{AD_1} \otimes \ket{0}_B) \label{eq:I_D_otimes_CNOT_AB}\\
        &= 
        &&\mkern-101mu\alpha \ket{\tau_0}_{D_1} 
        \ket{0}_A \ket{0}_B 
        + 
        \beta \ket{\tau_1}_{D_1} 
        \ket{1}_A \ket{1}_B \nonumber\\
        &=  
        \frac1{\sqrt2} \Big[
            &&\mkern-56mu\Big(
                \alpha \ket{\tau_{0}}_{D_1} \ket{0}_{A} 
                + 
                e^{-\theta_n i}\beta \ket{\tau_{1}}_{D_1}\ket{1}_{A} 
            \Big) 
            \ket{+_{\theta_n}}_{B} \nonumber\\
            &&&\mkern-78mu+\Big(
                \alpha \ket{\tau_{0}}_{D_1} \ket{0}_{A} 
                - 
                e^{-\theta_n i}\beta \ket{\tau_{1}}_{D_1}\ket{1}_{A}
            \Big) 
            \ket{-_{\theta_n}}_{B}
        \Big] \nonumber
        \\
        &=  
        \frac1{\sqrt2} 
        \Big[
            &&\mkern-56muR_Z(-\theta_n)_{A} \ket{\tau}_{AD_1} \ket{+_{\theta_n}}_{B}
            +
            R_Z(-\theta_n + \pi)_{A} \ket{\tau}_{AD_1} \ket{-_{\theta_n}}_{B}
        \Big] \nonumber\\
        \label{eq:post_CNOT_state}
        &=  
        \frac1{\sqrt2} 
        &&\mkern-61mu\sum_{r^B\in\{0,1\}}
            R_Z(-\theta_n+r^B\pi)_{A} \ket{\tau}_{AD_1} Z^{r^B}_B\ket{+_{\theta_n}}_{B} 
    \end{align}

    Note that register $B$ is then sent through a $Z^{c}$ gate before being measured in the $\theta_n$ basis giving classical outcome $x^B$.  From the point of view of the distinguisher, this is equivalent to performing the measurement first, writing the outcome $r^B$ and to define $x^B$ as $r^B \oplus c$.
    In this equivalent situation, one can assume without loss of generality that the measurement of $B$ giving $r^B$ is performed before the one of $A$ performed by the distinguisher. \Cref{eq:post_CNOT_state} shows that this first measurement yields a uniformly random outcome ($r^B$) and collapses the state of registers $A, D_1$ onto
    \begin{equation*}
        R_Z(-{\theta_n} + r^B\pi)_{A} \ket{\tau}_{A D_1}
    \end{equation*}

    In this case, the mixed state between steps 4 and 5 of the distinguisher's behavior (after the measurement performed by $\SD^B$ and after $\M^{(c)}$) is
    \begin{equation*}
        \psi^B
        =
        \sum_{\substack{r^B\in\{0,1\} \\ c\in\{0,1\}}}
        \frac{1}{p_c} \mathcal{E}_c \Big(
        R_Z(-{\theta_n} + r^B\pi)_{A} 
               \kb{\tau}_{A, D_1}   R_Z(-{\theta_n} + r^B\pi)_{A}^\dagger 
        \Big)
        \otimes
        \ketbra{(c\oplus r^B, c)}_{D_2},
    \end{equation*}
    where $\mathcal{E}_{c}$ is the completely positive trace non-increasing map corresponding to outcome $c$ for $\M^{(c)}$ and $p_c$ is a normalization factor.

    Otherwise, if $\S = \SD^C$, the measurement $\M^{(c)}$ is by definition of $\SD^C$ applied to
    $
        R_Z(-{\theta_n} + r^C\pi)_{A}
        \ket{\tau}_{A, D_1}
    $
    such that 
    \begin{equation*}
        \psi^C
        =
        \sum_{\substack{r^C\in\{0,1\} \\ c\in\{0,1\}}}
        \frac{1}{p_c} \mathcal{E}_c \Big(
        R_Z(-\theta_n + r^C\pi)_{A} 
        \kb{\tau}_{A, D_1}   R_Z(-\theta_n + r^C\pi)_{A}^\dagger 
        \Big)
        \otimes
        \ketbra{(c\oplus r^C, c)}_{D_2}.    
    \end{equation*}

    Since $\psi^B = \psi^C$, 
    the outcome $d$ of the final measurement $\M_\D$ does not depend on whether $\S$ is $\SD^C_n$ or $\SD^B_n\circ\SIM_n$. This concludes the proof.
\end{proof}

\Cref{thm:detector_single_qubit_systems} states that the state distribution $\SD^C_n$ of Charlie securely and composably implements the task of measuring a single qubit. The consequence is that, independently of the context, any such measurement can be implemented using a qubit rotation device realizing $\SD^C_n$ and an untrusted detector.

We are now ready to prove \Cref{thm:detector} which states that the state distribution $\SD^C$ of Charlie securely implements the one of the detector $B$.

\begin{theorem} \label{thm:detector}
    Under \Cref{ass:pb_structure},
    \begin{equation} \label{eq:thm:detector}
        \SD^C \approx_0 \SD^B \circ \SIM
    \end{equation}
\end{theorem}

\begin{proof}
    Composing in parallel \Cref{thm:detector_single_qubit_systems} for all $n\in[N]$ and then sequentially with $\G_\Theta$ gives
    \begin{align*}
        \G_\Theta \circ {\big|\big|}_{n\in[N]}\SD^C_n 
        &\approx_0
        \G_\Theta \circ {\big|\big|}_{n\in[N]} 
        \Big( 
            \SD^B_n \circ \SIM_n
        \Big)
        \\
        \G_\Theta \circ 
        {\big|\big|}_{n\in[N]}\SD^C_n
        &\approx_0
        \G_\Theta \circ
        \Big({\big|\big|}_{n\in[N]} \SD^B_n \Big)
        \circ
        \Big({\big|\big|}_{n\in[N]} \SIM_n  \Big)\\
        \SD^C&\approx_0
        \SD^B\circ\SIM
    \end{align*}
    where the second equation is a natural consequence of the definition of composition and is sometimes referred to as the interchange law for composition. The last equation follows by definition of $\SD^C$(\Cref{eq:SDCn_def}), $\SD^B$ (\Cref{eq:SDBn_def}) and $\SIM$ (\Cref{eq:SIMn_def}).
\end{proof}

\section{Implementation of a qubit source} \label{sec:Source}

We now focus on a cryptographic protocol $P$ involving a qubit source. We call Alice, or $A$, a given entity, $P^A$ the local protocol followed by $A$, and we suppose that $P^A$ involves the preparation of single-qubit quantum states. 

Similarly to the case of qubit measurements, we consider an entity called Charlie, or $C$, who wishes to assume the role of $A$ in $P$ even though his only quantum ability is to receive $N$ qubits, apply to each qubit $n\in[N]$ a rotation $R_Z(\varphi_n)$ with chosen $\varphi_n\in[0, 2\pi[$ and to output the rotated qubits.
Charlie will rely on an external, untrusted source to prepare the required number of qubits and will run these qubits in his device so that he outputs quantum states that can be used in the protocol $P$.

From the nature of Charlie's operation, we can only hope to achieve this task if the desired states are parametrized by only one secret parameter $\theta\in[0, 2\pi[$. We thus choose to restrict our analysis to the case where $P^A$ only involves the preparation of $\ket{+_\theta}$ states \footnote{The vast majority of quantum cryptographic protocols admit an equivalent formulation where the source(s) only have to prepare $\ket{+_\theta}$ states. Examples are given \Cref{sec:examples}.}.
Formally, we now consider the following \Cref{ass:pa_structure} instead of \Cref{ass:pb_structure}.
\begin{assumption} \label{ass:pa_structure}
    $P^A$ can be decomposed in two distinct phases:
    \begin{enumerate}
        \item \textbf{State distribution:} $A$ samples $\theta = (\theta_n)_{n\in[N]} \in [0, 2\pi[^N$ according to a given probability density function $p_\Theta$. They then prepare and output a quantum state of the form $\bigotimes\limits_{n\in[N]}\ket{+_{\theta_n}}$.
        \item \textbf{Classical post-processing:} $A$ performs an arbitrary classical protocol with inputs $(\theta_n)_{n\in N}$.
    \end{enumerate}
\end{assumption}
Additionally, and specifically for this section, we assume the following property on the probability distribution of the angles $\theta_n$\footnote{While this assumption might appear restrictive at first glance, it is satisfied by most common quantum cryptographic protocols. It is functionally equivalent to stating that for any angle $\theta_n$, the qubits are sent in the $\ket{+_\theta}$ and $\ket{-_\theta}$ states with equal probability.}.
\begin{assumption}\label{ass:p_Theta_periodicity}
    The probability density function $p_\Theta$ which determines the sampling of $(\theta_n)_{n\in[N]}$ is $\pi$-periodic, meaning that for all $\theta\in[0, 2\pi[$, $p_\Theta(\theta) = p_\Theta(\theta+\pi)$.
\end{assumption}

We define $\P^A$ as the system that formalizes the protocol $P^A$. As for $\P^B$ in \Cref{sec:Detector}, $\P^A$ can be decomposed as $\P^A = \SD^A \circ \PP$ with $\SD^A$ a system implementing the state distribution and $\PP$ performing the post-processing. Following \Cref{ass:pa_structure}, $\SD^A$ is a simple "qubit source" system that outputs at a user interface the angles $(\theta_n)_{n\in[N]}$ (to be received by $\PP$), and at an outer interface the quantum state $\bigotimes\limits_{n\in[N]}\ket{+_{\theta_n}}$.\\

We now describe how Charlie, who does not have access to a trusted single-qubit source,  will assume the role of $A$ in protocol $P$. A simple way for Charlie to output the desired $\ket{+_\theta}$ states for a given $\theta\in[0, 2\pi[$ is to ask for qubits in the $\ket{+}$ state and apply the corresponding rotation $R_Z(\theta)$ to them. However, if the untrusted source providing the qubits to Charlie deviates from the expected behavior, Charlie might output a different state without realizing it, and thus would not effectively implement the ideal source $\SD^A$ of the protocol $P^A$. We will show that under additional assumptions on the protocol itself, the overall protocol $\P^C_{source} = \SD^C_{source} \circ \PP$ securely implements the desired $\P^A$, even if Charlie does not implement $\SD^A$.

We formalize Charlie's behavior in the system $\SD^C_{source}$.

\begin{definition}[$\SD^C_{source}$]
    $\SD^C_{source}$ is a system formalizing Charlie's state distribution. It is represented in \Cref{fig:SDCNOTandSDCsource}.
    As all previous systems, $\SD^C_{source}$ starts by sampling $\theta = (\theta_n)_{n\in[N]} \in [0, 2\pi[^N$ according to $p_\Theta$. Then, successively for $N$ rounds $n\in[N]$, it receives a qubit \emph{from an untrusted source}, applies the $R_Z(\theta_n)$ gate to it and outputs the qubit back at its outer interface. $\SD^C_{source}$ also outputs $(\theta_n)_{n\in[N]}$ at its user interface (to be received by the post-processing $\PP$).
\end{definition}

We finally define the local protocol of Charlie as the composition of his state distribution $\SD^C_{source}$ and the post-processing $\PP$ of the protocol $\P$.
$$\P^C_{source} = \SD^C_{source} \circ \PP.$$ 
Our goal is now to show that (under some assumptions) $\P^C_{source}$ securely implements $\P^A$.
More precisely, we will exhibit two variants of $P^A$, called $P^{CNOT}$ and $P^{EB}$ and show the following results:

\begin{theorem} \label{thm:untrusted_+}
    Under \Cref{ass:pa_structure,ass:p_Theta_periodicity},
    if $\P^{CNOT}$ $\eps$-securely implements $\P^{A}$ for a given $\eps\in[0, \frac{1}{2}]$,
    then $\P^C_{source}$ also $\eps$-securely implements $\P^A$.
\end{theorem}

\begin{theorem} \label{thm:untrusted_EB}
    Under \Cref{ass:pa_structure,ass:p_Theta_periodicity},
    if $\P^{EB}$ $\eps$-securely implements $\P^{A}$ for a given $\eps\in[0, \frac{1}{2}]$,
    then $\P^C_{source}$ also $\eps$-securely implements $\P^A$.
\end{theorem}

In the two following \Cref{sec:thm+,sec:thmEB}, we define the respective variants $\P^{CNOT}$ and $\P^{EB}$ and show the corresponding \Cref{thm:untrusted_+,thm:untrusted_EB}. Most interestingly, we later show  examples in \Cref{sec:examples} on how our theorems naturally apply to most of the common quantum cryptography protocols, effectively showing that Charlie can assume the role of the source in these protocols.

\subsection{Proof of Theorem~\ref{thm:untrusted_+}} \label{sec:thm+}

The motivation for the first variant $\P^{CNOT}$ comes from the fact that preparing a qubit in a desired state can be done by preparing an entangled pair and measuring one of the qubits in the corresponding basis. Pushing the equivalence one step further, such an entangled pair can be obtained from a qubit in the $\ket+$ state by applying a $CNOT$ together with a second qubit in the $\ket0$ state. To prove \Cref{thm:untrusted_+}, we will show that under special assumptions in the protocol that uses the source, the same idea can be used even if the $CNOT$ is applied to an {\em untrusted} qubit that is supposed to be in the $\ket{+}$ state.

We start by defining this new system $\SD^{CNOT}$ that will replace the source in the protocol.

\begin{definition}[$\SD^{CNOT}$] \label{def:SD+}
    $\SD^{CNOT}$ is represented \Cref{fig:SDCNOTandSDCsource}.
    It starts by sampling $(\theta_n)_{n\in[N]} \in [0, 2\pi[^N$ according to $p_\Theta$. Then, successively for $N$ rounds $n\in[N]$, the system performs the following tasks:
    \begin{enumerate}
        \item It receives a qubit $Q_1$ \emph{from an untrusted source}. 
        \item \label{item:CNOT_setp} It applies a $CNOT$ gate onto that qubit, with a freshly prepared qubit $Q_2$ in the $\ket0$ state as the target qubit of the gate.
        \item It measures qubit $Q_2$ in the $-\theta_n$ basis, calling the outcome $x_n$.
        \item It outputs qubit $Q_1$.
    \end{enumerate} 
    Finally, $\SD^{CNOT}$ outputs $(\theta'_n)_{n\in[N]}$ where for all $n\in[N]$, $\theta'_n = \theta_n + x_n\pi$ (modulo $2\pi$),  meaning that the output angles are "flipped" for the rounds where $\ket{-_{-\theta_n}}$ has been measured.
\end{definition}

\begin{figure}[!ht]
    \centering 
    \includegraphics[width=0.85\textwidth]{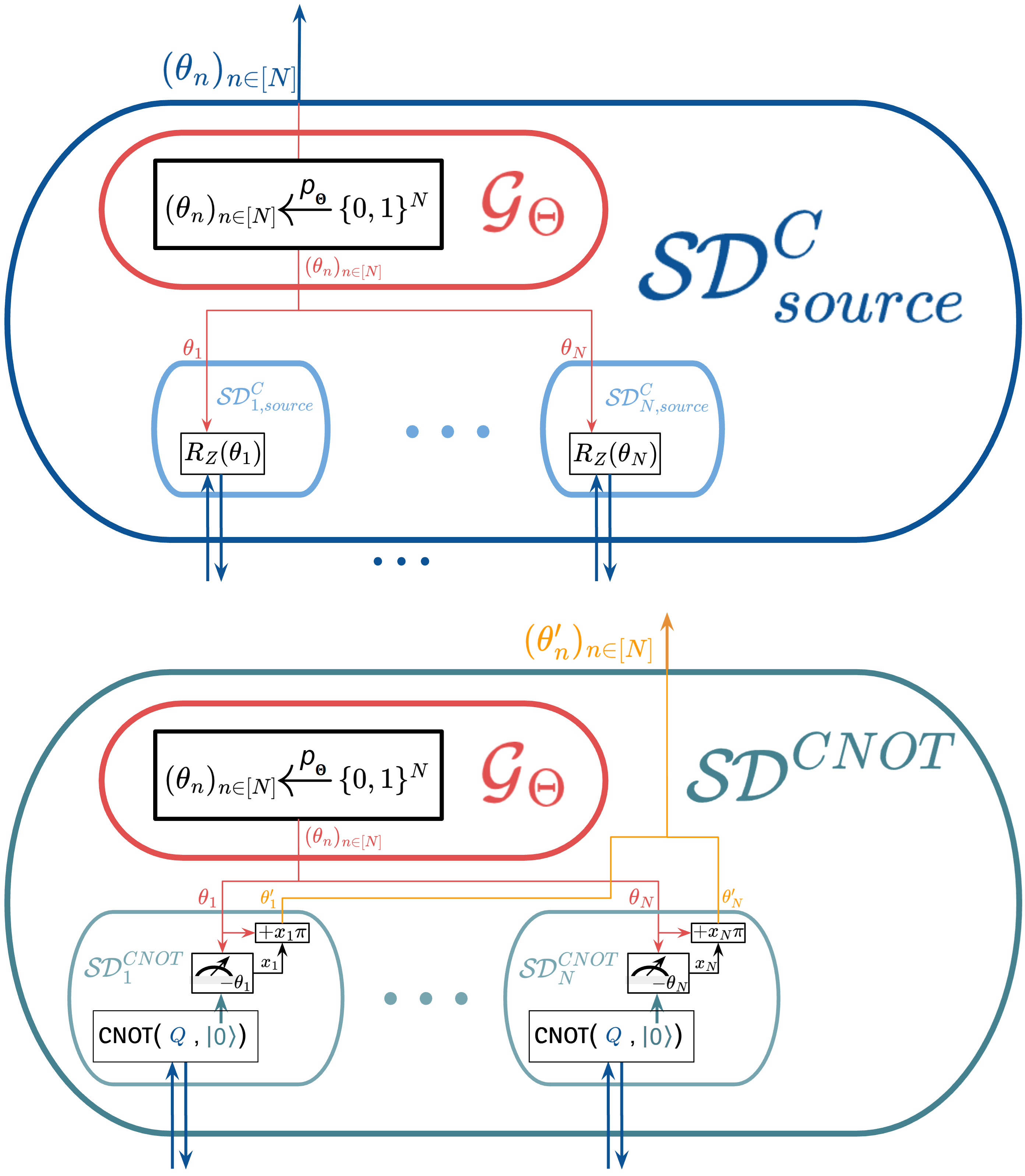}
    \caption{\centering Charlie's state distribution $\SD^{C}_{source}$ and its variant $\SD^{CNOT}$.}
    \label{fig:SDCNOTandSDCsource}
\end{figure}

In order to prove \Cref{thm:untrusted_+}, we first show that the state distribution $\SD^C_{source}$ of Charlie securely implements $\SD^{CNOT}$.

\begin{lemma} \label{thm:SDC=SD+}
    \begin{equation} \label{eq:SDC=SD+}
        \SD^C_{source} \approx_0 \SD^{CNOT}.
    \end{equation}
\end{lemma}

\begin{proof}
    As for the proof of \Cref{thm:detector_single_qubit_systems}, this proof is inspired by \cite{SecretSharingQline}.
    
    First note that from the way $\SD^C_{source}$ and $\SD^{CNOT}$ sample their angles and treat separate rounds $n\in[N]$, we can write:
    \begin{align} \label{eq:SDC-def-Gtheta}
        \SD^C_{source} &=  \G_\Theta \circ \Big({\big|\big|}_{n\in[N]} \SD^C_{n, source} \Big)\\
        \label{eq:SD+-def-Gtheta}
        \SD^{CNOT} &=  \G_\Theta \circ \Big({\big|\big|}_{n\in[N]} \SD^{CNOT}_{n} \Big)
    \end{align}
    where $\G_\Theta$ defined \Cref{eq:SDBn_def}, and for any $n\in[N]$, $\SD^C_{n, source}$ and $\SD^{CNOT}_{n}$ single round systems.
   
    Note further that, given \Cref{ass:p_Theta_periodicity}, if any subset of the angles output by $\G_\Theta$ is "flipped" (i.e by adding $\pi$ modulo $2\pi$), the resulting output angles would keep the same probability distribution and it would be indistinguishable from the original one (without the flips). 
    In particular, by defining $\F$ as a system that receives an angle $\theta\in[0, 2\pi[$ and randomly flips it by randomly sampling a bit $r$ and outputting $\theta+r\pi$, we get 
    \begin{equation}
        \G_\Theta \approx_0  \G_\Theta \circ \Big({\big|\big|}_{n\in[N]} \F \Big),
    \end{equation}
    which together with \Cref{eq:SDC-def-Gtheta} gives
    \begin{align} 
        \SD^C_{source} &\approx_0  \G_\Theta \circ \Big({\big|\big|}_{n\in[N]} \F \Big) \circ
        \Big({\big|\big|}_{n\in[N]} \SD^C_{n, source} \Big)\\
        &\approx_0\G_\Theta \circ \Big({\big|\big|}_{n\in[N]} \F \circ \SD^C_{n, source}\Big). \label{eq:SDC_from_FcircSDCn}
    \end{align}

    Note that if for all $n\in[N]$
    \begin{equation} \label{eq:SDCn=SD+n}
        \F \circ \SD^C_{n, source} \approx_0 \SD^{CNOT}_{n}
    \end{equation}
    then by composition the right-hand sides of \Cref{eq:SDC_from_FcircSDCn,eq:SD+-def-Gtheta} coincide, yielding \Cref{eq:SDC=SD+} and concluding the proof.
    The remainder of the proof is dedicated to the proof of \Cref{eq:SDCn=SD+n}. 
    
    Let $n\in[N]$ be fixed and consider a distinguisher which is given black-box access to a system $\S \in \{\F \circ \SD^C_{n,source},\SD^{CNOT}_{n}\}$ and whose goal is to distinguish the two cases.

    Without loss of generality, we can consider that the behavior of the distinguisher amounts to
    \begin{enumerate}
        \item Choosing an angle $\theta_n$ and preparing a pure state $\ket\tau_{A,D}$ in qubit register $A$, and a private register $D$.
        \item Sending $\theta_n$ and register $A$ to $\S$ for it to apply its operation.
        \item Getting back register $A$ from $\S$, as well as the announced angle $\theta'_n$ in a register $T$.
        \item Applying a final measurement $\M_\D$ of single bit output $d$ on registers $A$, $D$, $T$.
    \end{enumerate}

    \noindent We denote by $\psi^{CNOT}$ and $\psi^C$ the states of registers $A, D$ between steps 3 and 4 if $\S = \SD^{CNOT}_{n}$ or $\S = \F \circ \SD^C_{n,source}$, respectively. We will show that $\psi^{CNOT} = \psi^C$, thus proving that the measurement $\M_\D$ at step 4 cannot distinguish the two states better than a random guess.

    If $\S$ is $\F \circ \SD^C_{n,source}$, then by definition
    \begin{equation} \label{eq:psiC_result}
        \psi^{C} = \sum_{r\in\{0,1\}}R_Z(\theta_n+r\pi)_{A} \ketbra{\tau}_{AD} R_Z(\theta_n+r\pi)_{A}^\dagger \otimes \ketbra{\theta_n+r\pi}_T
    \end{equation}
    where the sum over $r\in\{0,1\}$ comes from the random angle flip of system $\F$.

    We now suppose that $\S = \SD^{CNOT}_{n}$ and write 
    $\alpha \ket{\tau_{0}}_D \ket{0}_{A} + \beta \ket{\tau_{1}}_D \ket{1}_{A} := \ket\tau_{AD}$. We call $B$ the register containing the target qubit of the $CNOT$ operation applied by $\SD^{CNOT}_{n}$ (the qubit called $Q_2$ in \Cref{def:SD+}).
    The state of registers $A, B, D$ after the $CNOT$ gate is
    \begin{align*}
        \ket{\psi^{CNOT}}_{ADB}
        =
        (\mathds{I}_D \otimes CNOT_{AB})
        (\ket\tau_{AD} \otimes \ket{0}_B)\\
    \end{align*}

    Up to the slight notational adjustment of associating the register names $D_1$ and $D$, this corresponds to  \Cref{eq:I_D_otimes_CNOT_AB}\footnote{The recurrence of this equation stems from the strong similarities between $\SD^{CNOT}_{n}$ (studied here) and $\SD^B_n \circ \SIM_n$ (studied in the proof of \Cref{thm:detector_single_qubit_systems})}. Similar algebraic manipulations to those that yielded \Cref{eq:post_CNOT_state} lead to 
    \begin{align*}
        \ket{\psi^{CNOT}}_{ADB}
        =  
        \frac1{\sqrt2} 
        \sum_{x^B\in\{0,1\}}
            R_Z(\theta_n+x^B\pi)_{A} \ket{\tau}_{AD} Z^{x^B}_B\ket{+_{-\theta_n}}_{B}.
    \end{align*}

    As a consequence, the measurement by $\S = \SD^{CNOT}_{n}$ of register $B$ in the $-\theta_n$ basis gives a uniformly random classical outcome $x^B$ and collapses the state of registers $A, D$ onto $R_Z(\theta_n+x^B\pi)_{A} \ket{\tau}_{AD}$. 
    
    Thus, since $\SD^{CNOT}_{n}$ announces the angle $\theta_n' = \theta_n+x^B\pi$,
    \begin{align*}
        \psi^{CNOT} &= \frac12\sum_{x^B\in\{0, 1\}} R_Z(\theta_n+x^B\pi)_{A} \ketbra{\tau}_{AD} R_Z(\theta_n+x^B\pi)_{A}^\dagger
        \otimes \ketbra{\theta_n+x^B\pi}
    \end{align*}
    which is equal to $\psi^{C}$ (see \Cref{eq:psiC_result}).

    As a result, $\M_\D$ and thus the distinguisher $\D$ cannot distinguish $\F \circ \SD^C_{n, source}$ from $\SD^{CNOT}_{n}$ better than a random guess, and \Cref{eq:SDCn=SD+n} follows.
\end{proof}

Similarly to the protocol $P^A$ which is defined as $\SD^A \circ \PP$, we define the variant protocol $\P^{CNOT}$ as $\SD^{CNOT} \circ \PP$ ( $\SD^{CNOT}$ being defined in \Cref{def:SD+}.
Using \Cref{thm:SDC=SD+}, we can now conclude the proof of \Cref{thm:untrusted_+}.
\begin{proof}[Proof of \Cref{thm:untrusted_+}]
The assumption of \Cref{thm:untrusted_+} on $\P^{CNOT}$ indeed implies that there exists a simulator $\SIM^{CNOT}$ such that $\P^{CNOT} \approx_\eps \P^{A} \circ \SIM^{CNOT}$.
Furthermore, composing the post processing system onto both sides of \Cref{eq:SDC=SD+} gives $\P^C_{source} \approx_0 \P^{CNOT}$.
As a consequence, $\P^C_{source} \approx_\eps \P^{A} \circ \SIM^{CNOT}$, meaning $\P^C_{source}$ $\eps$-securely implements $\P^A$. This concludes the proof of \Cref{thm:detector}.
\end{proof}

\subsection{Proof of Theorem \ref{thm:untrusted_EB}} \label{sec:thmEB}

In this section, we define and prove the entanglement-based variant of \Cref{thm:untrusted_+}, namely \Cref{thm:untrusted_EB},
Specifically, while \Cref{thm:untrusted_+} requires that $\P^A$ admits a secure implementation wherein the state distribution consists of entangling untrusted qubits with locally prepared ones and working with the obtained entangled pairs, \Cref{thm:untrusted_EB} requires a stricter condition. 
In this case, we assume a secure implementation of $\P^A$ where the qubit pairs are entirely untrusted, giving more power to the distinguisher.
\Cref{thm:untrusted_EB} is a strict implication of~\Cref{thm:untrusted_+}; it can be seen as an operational specialization.
As demonstrated in \Cref{sec:examples}, despite its stronger requirements, \Cref{thm:untrusted_EB} remains broadly applicable to standard quantum cryptographic protocols, which regularly work with entangled qubit pairs assumed to be prepared directly by the adversary.

We define below the variant $\SD^{EB}$ of the state distribution. It is similar to $\SD^{CNOT}$, but instead of receiving an untrusted qubit that is supposed to be a $\ket+$ state and entangling it with another qubit,  $\SD^{EB}$ directly receives a pair of qubits that are supposed to be an EPR pair.

\begin{definition}[$\SD^{EB}$]
    $\SD^{EB}$ is called the "entanglement-based" variant of $\SD^A$. As $\SD^A$, it starts by sampling $(\theta_n)_{n\in[N]} \in [0, 2\pi[^N$ according to $p_\Theta$. Then, successively for $N$ rounds $n\in[N]$, the system performs the following tasks:
    \begin{enumerate}
        \item It receives 2 qubits \emph{from an untrusted source}. 
        \item It measures one qubit in the $-\theta_n$ basis, calling the outcome $x_n$.
        \item It outputs the second qubit.
    \end{enumerate} 
    Finally, $\SD^{EB}$ outputs $(\theta'_n)_{n\in[N]}$ where for all $n\in[N]$, $\theta'_n = \theta_n + x_n\pi$ (modulo $2\pi$),  meaning that the output angles are "flipped" for the rounds where $\ket{-_{-\theta_n}}$ has been measured.
\end{definition}

We can then define $\P^{EB}=\SD^{EB}\circ\PP$, and we can prove \Cref{thm:untrusted_EB} that we stated previously.

\begin{proof}[Proof of \Cref{thm:untrusted_EB}]
    Note that $\SD^{EB}$ adopts the same behavior on the qubit pairs as the previously defined $\SD^{CNOT}$ after its $CNOT$ operation.
    As a consequence, by writing $\R$ a system receiving single qubits and performing step \ref{item:CNOT_setp} of \Cref{def:SD+}, 
    \begin{equation}
        \SD^{CNOT} = \SD^{EB} \circ \R
    \end{equation}
    which, by composition with the post-processing system $\PP$, gives \begin{equation}\label{eq:SDEB_def_from_SD+}
        \P^{CNOT} = \P^{EB} \circ \R.
    \end{equation}

    Considering the assumption of \Cref{thm:untrusted_EB} on $\P^{EB}$ which implies that there exists a simulator $\SIM^{EB}$ such that $\P^{EB} \approx_\eps  \P^{A} \circ \SIM^{EB}$, \Cref{eq:SDEB_def_from_SD+} gives $\P^{CNOT} \approx_\eps \P^{A} \circ \SIM^{EB} \circ \R$, meaning that $\P^{CNOT}$ $\eps$-securely implements $\P^{A}$.
    \Cref{thm:untrusted_+} concludes the proof.
\end{proof}

\section{Applications} \label{sec:examples}

We now describe examples of protocols that meet the requirement of \Cref{thm:detector,thm:untrusted_EB}, and thus for which the trusted preparations and measurements can be securely replaced by untrusted ones followed by a rotation.

\subsection{Quantum Key Distribution}
The crown jewel of quantum cryptography, Quantum Key Distribution (QKD) protocols, allow two parties Alice and Bob to establish a shared secret using quantum communication devices based only on authenticated classical channels. 
The most studied variants are entanglement-based QKD~\cite{Portman_Renner14,TL17} and prepare-and-measure (BB84) QKD~\cite{TL17}.

In entanglement-based QKD, we only require trusted measurements and \Cref{thm:detector} thus implies that both players' local protocol instances---the systems which model their individual input-output behavior within the broader protocol---can equivalently be implemented using the device of Charlie. 

In prepare-and-measure QKD, however, Alice must send qubits to Bob, prepared by a trusted process, who then measures them also using a trusted detector. While we can again use \Cref{thm:detector} to replace Bob's device straightforwardly, we argue in the following that most prepare-and-measure protocols also satisfy the requirement for \Cref{thm:untrusted_EB} to apply, leading to the security of an implementation of the state distribution of Alice using $\SD^C_{source}$.

We note that in most common prepare-and-measure QKD protocols, qubits are sent in random states among the eigenstates of the $Z$ or $X$ basis, and thus these protocols satisfy \Cref{ass:p_Theta_periodicity,ass:pa_structure} up to the appropriate basis change. 

In the situation where $\P^A$ is Alice's local protocol instance of prepare-and-measure QKD, $\P^{EB}$ is a system that receives qubit pairs, measures one qubit of each pair (uniformly at random in one of two mutually unbiased bases), outputs the other qubits and proceeds to QKD's post-processing. The only difference between $\P^{EB}$ and Alice's local protocol instance of entanglement-based QKD is thus the fact that for each received qubit, $\P^{EB}$ receives one more qubit and directly outputs it without further treatment. These two systems are thus equivalent ($\approx_0$) up to a simulator (on entanglement-based QKD's side) that receives qubits and immediately outputs them. As a result, $\P^{EB}$ is a secure implementation of Alice's local protocol instance of entanglement-based QKD.

Since Alice's behavior is the only difference between the whole protocols of entanglement-based and prepare-and-measure QKD, Alice's local protocol instance of entanglement-based QKD and thus $\P^{EB}$ securely implements $\P^A$.

As a consequence, \Cref{thm:untrusted_EB} applies to most prepare-and-measure QKD protocols, thus showing that Charlie can securely assume the role of Alice using $\P^C_{source}$. We note that, combined with \Cref{thm:detector} this achieves a concurrent proof for the existence of a (composably) secure protocol for QKD on the Qline, a result previously established by Grilo \textit{et al.} \cite{SecretSharingQline}.

\subsection{Quantum Oblivious Transfer}
Quantum Oblivious Transfer (QOT) is a cryptographic protocol between a sender (Alice) and a receiver (Bob)~\cite{QOT}. Traditionally, QOT involves Alice sending single qubits each in a random state among $\ket{0}$, $\ket{1}$, $\ket{+}$ or $\ket{-}$, and Bob measuring these qubits. Alice and Bob are thus naturally implemented using a qubit source and qubit detector, respectively.
This protocol satisfies (up to an appropriate basis change on the states of the qubits) \Cref{ass:pb_structure,ass:pa_structure,ass:p_Theta_periodicity}. As a consequence, \Cref{thm:detector} implies that Charlie can implement the local protocol instance of the receiver with $\SD^C$.
We argue in the following that Oblivious Transfer also satisfies the requirement for \Cref{thm:untrusted_EB} to apply, leading to the security of an implementation of the state distribution of Alice using $\SD^C_{source}$.

\medskip

In order to show that \Cref{thm:untrusted_EB} applies, we will refer to the security proof of \cite{bf12}. On a high level, the security of QOT is proven by equivalence to an "entanglement-based" version of the protocol named QOT$^*$. In this version, Alice prepares EPR pairs, sends one qubit of each pair to Bob, and proceeds to the post-processing of standard Quantum OT, checking Bob's commitments by comparison with the measurement outcomes of the qubits she kept.
The security of QOT$^*$ is proven by showing that the measurements of Alice can be seen as a sampling strategy, which collapses the state of Alice's remaining qubits onto a particular form, and from which we can bound the entropy of Bob on Alice's secret.

We note that when formalized this way, by defining $\P^A$ as Alice's protocol in the original, prepare-and-measure QOT, $\P^{EB}$ is almost exactly Alice's protocol for QOT$^*$. Formally, up to a basis change, the differences are:
\begin{itemize}
    \item A delayed measurement (which commutes with Bob's possible actions and thus does not change the proof).
    \item The preparation of the EPR pairs, which is trusted in QOT$^*$ and not in $\P^{EB}$.
\end{itemize}

We remark, however, that while the authors of this work considered such a trusted preparation setup for simplicity, their security proof of QOT$^*$ still holds in the case of an untrusted preparation (as in $\P^{EB}$). Indeed, the quantum state shared between Alice and Bob is treated as an arbitrary state (potentially created by the adversary), and the collapsing effect of the sampling strategy \emph{does not depend on the initial state of the qubits}.
As a consequence, $\P^{EB}$ also securely implements QOT$^*$ and equivalently QOT. \Cref{thm:untrusted_EB} thus shows that Charlie can securely assume the role of Alice in Quantum Oblivious Transfer using $\P^C_{source}$.

\subsection{Other Protocols}

Following the two primitives studied above, we conjecture that several other protocols fall within the scope of application of our framework and admit a secure implementation using single-qubit rotation devices instead of single-qubit sources or detectors. In particular, we expect \Cref{thm:detector} (concerning qubit measurements) to apply to the vast majority of protocols, and our \Cref{thm:untrusted_EB} to likely apply to protocols that incorporate some sort of self-testing or consistency checking phase like the commit-and-reveal step of OT or the tests of QKD.
Additionally, one could consider artificially adding such a testing step to protocols that do not directly fit our framework in order to derive a modified protocol that admits, as a consequence of either \Cref{thm:untrusted_+} or \Cref{thm:untrusted_EB}, a secure implementation based on single-qubit rotation devices.

Among the most promising candidates, we identify the following primitives:

\begin{itemize}
    \item \textbf{Quantum Bit commitment (QBC): }
    QBC is a cryptographic primitive between a committer and a receiver for which protocols have been proven secure in the Noisy Storage and Bounded Storage Models~\cite{QBC_BSM,QBC_NSM,QBC_BNSM,grilo2024roundcomplexity}. It has been studied under both prepare-and-measure scenarios, where either the committer or the receiver requires a qubit source while the other party uses a measurement device and we conjecture that this can be used to show that \Cref{thm:untrusted_EB} applies. 
    \item \textbf{Delegated Quantum Computing (DQC): }
    DQC is a primitive where a client wants to delegate a (quantum) computation to an untrusted quantum computer. In the most standard protocol, the client sends encrypted single qubit states that the server then uses to build a graph state and compute the desired task in the Measurement-Based Quantum Computing framework. 
    In fact, Polacchi \textit{et al.}~\cite{polacchi2023multi} have already introduced a secure protocol for (multi-client) DQC using a single-qubit rotation device on the client side. Showing that our theorem applies to the original protocol would effectively consist of an alternative security proof for such a DQC protocol where the client uses a single-qubit rotation device instead of a single-qubit source.
\end{itemize}

\newpage
\bibliography{MyRef.bib}

@INPROCEEDINGS{doosti-hanouz23establishing,
  author={Doosti, Mina and Hanouz, Lucas and Marin, Anne and Kashefi, Elham and Kaplan, Marc},
  booktitle={2024 International Conference on Quantum Communications, Networking, and Computing (QCNC)}, 
  title={Establishing Shared Secret Keys on Quantum Line Networks: Protocol and Security}, 
  year={2024},
  volume={},
  number={},
  pages={176-183},
  doi={10.1109/QCNC62729.2024.00035}}

@article{polacchi2023multi,
  title={Multi-client distributed blind quantum computation with the Qline architecture},
  author={Polacchi, Beatrice and Leichtle, Dominik and Limongi, Leonardo and Carvacho, Gonzalo and Milani, Giorgio and Spagnolo, Nicol{\`o} and Kaplan, Marc and Sciarrino, Fabio and Kashefi, Elham},
  journal={Nature Communications},
  volume={14},
  number={1},
  pages={7743},
  year={2023},
  publisher={Nature Publishing Group UK London}
}

@article{CPEW17,
	author = {Marco Clementi and Anna Pappa and Andreas Eckstein and Ian A. Walmsley and Elham Kashefi and Stefanie Barz},
	journal = {Phys. Rev. A},
	number = {062317},
	title = {Classical multiparty computation using quantum resources},
	volume = {96},
	year = {2017}}

@inproceedings{AbstractCrypto,
	author = {Ueli Maurer and Renato Renner},
	booktitle = {Proceedings of Innovations in Computer Science, ICS 2010, Tsinghua University Press},
	pages = {1-21},
	title = {Abstract Cryptography},
	year = {2011}}

@article{TL17,
   title={A largely self-contained and complete security proof for quantum key distribution},
   volume={1},
   ISSN={2521-327X},
   url={http://dx.doi.org/10.22331/q-2017-07-14-14},
   DOI={10.22331/q-2017-07-14-14},
   journal={Quantum},
   publisher={Verein zur Forderung des Open Access Publizierens in den Quantenwissenschaften},
   author={Tomamichel, Marco and Leverrier, Anthony},
   year={2017},
   month=jul, 
   pages={14} 
}

@article{BB84,
title = {Quantum cryptography: Public key distribution and coin tossing},
journal = {Theoretical Computer Science},
volume = {560},
pages = {7-11},
year = {2014},
issn = {0304-3975},
doi = {https://doi.org/10.1016/j.tcs.2014.05.025},
url = {https://www.sciencedirect.com/science/article/pii/S0304397514004241},
author = {Charles H. Bennett and Gilles Brassard}
}

@inproceedings{BF12,
  title={Sampling in a quantum population, and applications},
  author={Bouman, Niek J and Fehr, Serge},
  booktitle={Annual Cryptology Conference},
  year={2010},
  organization={Springer}
}

@article{Portman_Renner14,
  title={Cryptographic security of quantum key distribution},
  author={Christopher Portmann and Renato Renner},
  journal={ArXiv},
  year={2014},
  volume={abs/1409.3525},
  url={https://api.semanticscholar.org/CorpusID:14053576}
}

@article{SecretSharingQline,
      title={Security of a secret sharing protocol on the Qline}, 
      author={Alex B. Grilo and Lucas Hanouz and Anne Marin},
      year={2025},
      journal={ArXiv},
      volume={abs/2504.19702},
      eprint={2504.19702},
      archivePrefix={arXiv},
      primaryClass={quant-ph},
      url={https://arxiv.org/abs/2504.19702}, 
}

@ARTICLE{Qline_implementation_berlin,
  author={Sena, Matheus and Harder, Georg and Döring, Ronny and Braun, Ralf-Peter and Ritter, Michaela and Holschke, Oliver and Kaplan, Marc and Geitz, Marc},
  journal={IEEE Photonics Journal}, 
  title={Deploying the Qline System for a QKD Metropolitan Network on the Berlin OpenQKD Testbed}, 
  year={2025},
  volume={17},
  number={1},
  pages={1-11},
  doi={10.1109/JPHOT.2024.3516138}}

@article{Schmid_2005,
   title={Experimental Single Qubit Quantum Secret Sharing},
   volume={95},
   ISSN={1079-7114},
   url={http://dx.doi.org/10.1103/PhysRevLett.95.230505},
   DOI={10.1103/physrevlett.95.230505},
   number={23},
   journal={Physical Review Letters},
   publisher={American Physical Society (APS)},
   author={Schmid, Christian and Trojek, Pavel and Bourennane, Mohamed and Kurtsiefer, Christian and Żukowski, Marek and Weinfurter, Harald},
   year={2005},
   month=dec }

@inproceedings{QOT,
author = {Bennett, Charles and Brassard, Gilles and Crépeau, Claude and Skubiszewska, Marie-Hélène},
year = {1991},
month = {01},
pages = {351-366},
title = {Practical Quantum Oblivious Transfer.},
volume = {576},
booktitle = {Annual international cryptology conference},
publisher = {Springer Berlin Heidelberg},
}

@article{Qenclave,
   title={QEnclave - A practical solution for secure quantum cloud computing},
   volume={8},
   ISSN={2056-6387},
   url={http://dx.doi.org/10.1038/s41534-022-00612-5},
   DOI={10.1038/s41534-022-00612-5},
   number={1},
   journal={npj Quantum Information},
   publisher={Springer Science and Business Media LLC},
   author={Ma, Yao and Kashefi, Elham and Arapinis, Myrto and Chakraborty, Kaushik and Kaplan, Marc},
   year={2022},
   month=Nov }

@misc{KLMO24,
      title={Verification of Quantum Computations without Trusted Preparations or Measurements}, 
      author={Elham Kashefi and Dominik Leichtle and Luka Music and Harold Ollivier},
      year={2024},
      eprint={2403.10464},
      archivePrefix={arXiv},
      primaryClass={quant-ph},
      url={https://arxiv.org/abs/2403.10464}, 
}

@article{QBC_BSM,
author = {Damg\r{A}rd, Ivan B. and Fehr, Serge and Salvail, Louis and Schaffner, Christian},
title = {Cryptography in the Bounded-Quantum-Storage Model},
journal = {SIAM Journal on Computing},
volume = {37},
number = {6},
pages = {1865-1890},
year = {2008},
doi = {10.1137/060651343}
}

@article{QBC_NSM,
  title = {Cryptography from Noisy Storage},
  author = {Wehner, Stephanie and Schaffner, Christian and Terhal, Barbara M.},
  journal = {Phys. Rev. Lett.},
  volume = {100},
  issue = {22},
  pages = {220502},
  numpages = {4},
  year = {2008},
  month = {Jun},
  publisher = {American Physical Society},
  doi = {10.1103/PhysRevLett.100.220502},
  url = {https://link.aps.org/doi/10.1103/PhysRevLett.100.220502}
}

@article{QBC_BNSM,
   title={Unconditional Security From Noisy Quantum Storage},
   volume={58},
   ISSN={1557-9654},
   url={http://dx.doi.org/10.1109/TIT.2011.2177772},
   DOI={10.1109/tit.2011.2177772},
   number={3},
   journal={IEEE Transactions on Information Theory},
   publisher={Institute of Electrical and Electronics Engineers (IEEE)},
   author={Konig, Robert and Wehner, Stephanie and Wullschleger, Jürg},
   year={2012},
   month=Mar, pages={1962–1984} }

@misc{grilo2024roundcomplexity,
      title={The Round Complexity of Proofs in the Bounded Quantum Storage Model}, 
      author={Alex B. Grilo and Philippe Lamontagne},
      year={2024},
      eprint={2405.18275},
      archivePrefix={arXiv},
      primaryClass={quant-ph},
      url={https://arxiv.org/abs/2405.18275}, 
}

\end{document}